\documentclass{acm_proc_article-sp}
\pdfoutput=1
\usepackage{latexsym}
\usepackage{booktabs}
\usepackage{smallsec}
\usepackage{amsmath,amssymb}
\usepackage[boxed]{algorithm}
\usepackage{algorithmic}
\usepackage{graphicx}
\usepackage{multirow}
\usepackage[tight]{subfigure}

\usepackage[pdftex,naturalnames,pdfpagemode={UseOutlines},letterpaper,bookmarks,bookmarksopen=true,pdfstartview={FitH},bookmarksnumbered=true,pdftitle={-},pdfauthor={-}]{hyperref}

\renewcommand{\paragraph}[1]{\vspace*{1mm}\noindent{\bf #1}}
\newcommand{\eat}[1]{}

\newtheorem{definition}{Definition}[section]

\newtheorem{theorem}[definition]{Theorem}

\newtheorem{example}[definition]{Example}

\newcommand{\squishlist}{
\begin{list}{$\bullet$}
{ \setlength{\itemsep}{0pt} \setlength{\parsep}{3pt}
\setlength{\topsep}{3pt} \setlength{\partopsep}{0pt}
\setlength{\leftmargin}{1.5em} \setlength{\labelwidth}{1em}
\setlength{\labelsep}{0.5em} } }

\newcommand{\squishend}{
\end{list}  }

\newcounter{Lcount}
\newcommand{\squishlisttwo}{
\begin{list}{\arabic{Lcount}.}
{ \usecounter{Lcount}
\setlength{\itemsep}{0pt}
\setlength{\parsep}{0pt}
\setlength{\topsep}{0pt}
\setlength{\partopsep}{0pt}
\setlength{\leftmargin}{2em}
\setlength{\labelwidth}{1.5em}
\setlength{\labelsep}{0.5em} } }

\newcommand{\squishlisttwob}{
\begin{list}{$\bullet$}
{ \setlength{\itemsep}{0pt} \setlength{\parsep}{0pt}
\setlength{\topsep}{0pt} \setlength{\partopsep}{0pt}
\setlength{\leftmargin}{2em} \setlength{\labelwidth}{1.5em}
\setlength{\labelsep}{0.5em} } }

\begin{document}

\title{Handling Skew in Multiway Joins in Parallel Processing}

\numberofauthors{1}
\author{Foto N. Afrati$^{\dag}$,
Jeffrey D. Ullman$^{\ddag}$,
Angelos Vasilakopoulos$^{\dag}$\\
\affaddr{\large $^{\dag}$National Technical University of Athens,
$^{\ddag}$Stanford University}\\
\email{afrati@softlab.ece.ntua.gr, avasilako@gmail.com, ullman@gmail.com}}

\date{}

\maketitle

{\bf Abstract:}
Handling skew is one of the major challenges in  query processing. In distributed
computational environments such as MapReduce,  uneven distribution
of the data to the servers is not desired. One of the dominant measures that
we want to optimize in distributed environments is communication cost. In a MapReduce job
this is the amount of data that is transferred from the mappers to the reducers.
In this paper we will introduce a novel technique for handling skew when we want to compute a multiway join in one MapReduce round with  minimum communication cost.
This technique is actually  an adaptation of the Shares algorithm~\cite{AfUl}.

\section{Introduction}

Systems such as Pig or Hive that implement SQL or relational algebra over MapReduce have mechanisms to deal with joins where there is significant skew; i.e.,
certain values of the join attribute(s) appear very frequently
(see, e.g.,~\cite{hive,pig,
Gates09}.
These systems use a two-round algorithm, where the first round identifies the {\em heavy hitters (HH)}, those values of the join attribute(s) that occur in at least some given fraction of the tuples. In the second round, tuples that do not have a heavy-hitter for the join attribute(s) are handled normally. That is, there is one reducer\footnote{In this paper, we use the term {\em reducer} to mean the application of the Reduce function to a key and its associated list of values. It should not be confused with a Reduce task, which typically executes the Reduce function on many key and their associated values.}
for each key, which is associated with a value of the join attribute(s).
 Since the key is not a heavy hitter, this reducer handles only a small fraction of the tuples, and thus will not cause a problem of skew.
For tuples with heavy hitters,  new keys are created that are handled along with the other keys (normal or those for other heavy hitters) in a single MR job.
The new keys in
these systems are created with  a simple technique as in the following example:
\vspace*{-.3cm}
\begin{example}
\label{1-ex}

We have to compute the join $R(A,B) \bowtie S(B,C)$ using a given number, $k$, of reducers.
Suppose  value $b$ for attribute $B$ is identified as a heavy hitter.
Suppose there are $r$ tuples of $R$ with $B=b$ and there are $s$ tuples of $S$ with $B=b$. Suppose also for convenience that $r>s$.
The distribution to $k$ buckets/reducers is done in earlier approaches by partitioning the data of one of the relations in $k$ buckets (one bucket for each reducer) and sending the data of the other relation to all reducers. Of course since $r>s$, it makes sense to choose relation $R$ to partition. Thus
values of attribute $A$ are hashed to $k$ buckets, using a hash function $h$, and each tuple of relation $R$  with $B=b$ is sent to one reducer -- the one that corresponds to the bucket that the value of the first argument of the tuple was hashed. The tuples of $S$ are sent to all the $k$ reducers.
Thus the number of tuples transferred from mappers to reducers is $ r + ks$.
\end{example}

\vspace*{-.3cm}

The approach described above appears not only in Pig and Hive, but dates back to~\cite{WolfDY93}. The latter work, which looked at a conventional parallel implementation of join rather than a MapReduce implementation, uses the same (non-optimal) strategy of choosing one side to partition and the other side to replicate.
In particular, these techniques are not optimal with respect to  {\em communication cost} (i.e.,  the number of inputs
transferred from the mappers to the reducers~\cite{AfratiSSU13,SOCC}).

{\bf Our contribution:}
In Example~\ref{2-way-ex} we show how we  can do significantly better than the standard technique of Example~\ref{1-ex}.  In the rest of the paper we show how the idea
 in Example~\ref{2-way-ex}  can be
extended to apply on any multiway join and for any number of heavy hitters.
 In particular, we show how to adapt Shares algorithm \cite{AfUl} to find a solution
that minimizes communication cost in the case there are heavy hitters.

\vspace*{-.5cm}
\begin{example}
\label{2-way-ex}
We take again  the join $R(A,B) \bowtie S(B,C)$.
We partition the tuples of $R$ with $B=b$ into $x$ groups and we also partition the tuples of $S$ with $B=b$ into $y$ groups, where $xy = k$.
We use one of the $k$ reducers for each pair $(i, j)$ for a group $i$ from R and for a
group $j$ from $S$. Now we are going to
partition tuples from R and S and we use hash functions $h_r$ and $ h_s$ to do the partitioning.
We send each tuple $(a,b)$ of $R$ to all reducers of the form $(i, q)$,
where $i=h_r(a)$
 is the group in which tuple $(a,b)$ belongs and $q$ ranges over all $y$ groups. Similarly,  we send each tuple $(b,a)$ of $R$ to all reducers of the form $(q, i)$,
where $i=h_s(a)$
 is the group in which tuple $(b,a)$ belongs and $q$ ranges over all $x$ groups.
Thus each tuple with $B=b$ from $R$ is sent to $y$ reducers and each tuple with
$B=b$ from $S$ is sent to $x$ reducers. Hence the
communication cost is $ry + sx$. We can show (see \cite{AfUl}) that by minimizing $ry + sx$ under
the constraint $xy = k$ we achieve communication cost equal to $\sqrt{2krs}$, which is always
less than what we found in Example~\ref{1-ex} which was $ r + ks$. The proof is easy:  $\sqrt{2krs} \leq  r + ks$
or  $0 \leq  \sqrt{r/s} - \sqrt{2k} + k\sqrt{s/r}$, which is a second order polynomial wrto $\sqrt{k}$ as unknown and it is positive for any $k$. Moreover observe that the improvement is significant:
The optimal communication cost grows as $\sqrt{k}$, while $ r + ks$ grows linearly with $k$.
\end{example}
\vspace*{-.4cm}

{\bf Related Work}
There is a lot of work over the decades about how to handle skew when we process queries.
We will limit ourselves here to recent work that considers joins in MapReduce or
discusses the Shares algorithm.
In \cite{BeameKS13} it is proven that with high probability the  Shares algorithm distributes tuples evenly
on uniform databases (these are defined precisely in \cite{BeameKS13} to be databases which resemble the case of
random data). Then,
 \cite{BeameKS14}  generalizes and enhances  results in  \cite{BeameKS13} and
\cite{AfratiSSU13}.
\cite{BeameKS14}
 describes how the Shares algorithm behaves on skewed data: it shows that the algorithm is resilient
to skew, and gives an upper bound even on skewed databases.  However this resilience applies
to ordinary joins that use many of the attributes in one relation allowing thus the tuples with a heavy hitter
to be distributed in many reducers. However this is not the case in the 2-way join example we gave -- and many others.

\section{Shares Algorithm}
The algorithm is based on
a schema according to which we distribute the data to a given number of $k$
reducers. Each reducer is defined by a vector, where each component of the vector corresponds to an attribute.
The algorithm uses a number of independently chosen random hash functions $h_i$
one for
each attribute $X_i$. Each tuple is sent to a number of reducers depending on the value of $h_i$ for the
specific attribute $X_i$ in this tuple. If $X_i$ is not present in the tuple, then the tuple is sent to all
reducers for all $h_i$ values.
For an example, suppose we have the 3-way join $R_1(X_1,X_2) \bowtie
R_2(X_2,X_3) \bowtie R_3(X_3,X_1)$.
In this example each reducer is defined by a vector $(x,y,z)$.
A tuple $(a,b)$ of $R_1$ is sent to a number of reducers
and specifically to reducers
$(h_1(a), h_2(b), i)$ for all $i$. I.e., this tuple needs to be replicated a number of times, and specifically
in as many reducers as is the number of buckets into which $h_3$ hashes the values of attribute $X_3$.
\vspace*{-.1cm}

When the hash function $h_i$ hashes the values of attribute $X_i$ to $x_i$ buckets, we say that the {\em share} of $X_i$ is $x_i$.
The communication cost is calculated to be, for each relation, the size of the relation times the replication
that is needed for each tuple of this relation. This replication can be calculated to be the product
of the shares of all the attributes that do not appear in the relation.
In order to keep the number of reducers equal to $k$, we need to calculate the shares so that their product is equal to $k$.
\vspace*{-.1cm}

Thus, in our example, the communication cost is $r_1x_3+r_2x_1+r_3x_2$ and we must have $x_1x_2x_3=k$. (We denote the size of a relation $R_i$ by $r_i$.)
In \cite{AfUl}, it is explained how to use the Lagrangean method to find the shares that minimize the communication cost, for any multiway join.
\vspace*{-.1cm}

We are going to need an important observation that was proven in \cite{AfUl}.
An attribute $A$ is {\em dominated} by attribute $B$ in the join if $B$ appears in all relations where $A$ appears.
It is shown that if an attribute is dominated, then it does not get a share, or, in other words, its share is equal to 1.
\subsection{Our Setting}

We saw how to compute  the 2-way join  in Example~\ref{2-way-ex} for the tuples which have one HH.
For this join, we  took two sets of keys:
\squishlisttwob
 \item The set of keys as presented   in Example~\ref{2-way-ex}
which send tuples with HH to a number of reducers in order to compute the join of tuples with HH.
\item The set of keys which send tuples without HH to a number of reducers in order to compute the join of tuples without HH. This second set is formed exactly as in the Shares algorithm.
\squishend
\vspace*{-.2cm}
It is convenient to see these two sets of keys as corresponding to two joins which we call
{\em residual joins}, and which actually differ only on the subset of the data they are applied. One applies the original join on the data with HH and the
other  applies the original join on the data without HH.

The method we presented in Example~\ref{2-way-ex} is actually based on the Shares algorithm. To see this, we  can be
equivalently thought as: We replace
 each tuple of relation $R$ with a tuple where $B$ has distinct fresh values
$b_1,b_2,\ldots $ and the same for the tuples of relation $S$ with $B$ having values $b'_1,b'_2,\ldots$.
Now we can apply the Shares algorithm to find the shares and distribute the tuples to reducers normally. The only
problem with this plan is that the output will be empty because we have chosen $b_i$s and $b'_i$s
to be all distinct.
This problem however has an easy solution, because, we can keep this replacement
to the conceptual level, just so to create a HH-free join and be able to apply the Shares algorithm
and compute the shares optimally. When we transfer the tuples to the reducers, however, we transfer the
original tuples and thus, we produce the desired output. We explain this conceptual structure in Section~\ref{mod-shares-ex-sec}.

Our setting is as follows: We have $k$ reducers to use for computing all residual joins.
We assume each residual join $J_i$ uses $k_i$ of those reducers, thus one
constraint is $k_1+k_2+\cdots =k$.
For each residual join, we need to compute the communication cost expression. The objective function to minimize is the sum of the cost expressions over all residual joins, under the
constraint:  for each residual join $J_i$,  the product of the attribute shares must be equal to $k_i$.

The aim of this paper is to show how to systematically apply the idea explained for the 2-way join on
any multiway join with any number of HH.
The structure of the rest of the paper is the following:
\vspace*{-.2cm}
\squishlisttwo
\item We decompose into residual joins, i.e., we partition the data into subsets and we view a residual join as the
original join  applied~on~one~of~the~subsets~(Section 3).

\item We explain how to form a {\em HH-free residual join} and  how to compute the communication cost expression
for each residual join (Section 4).
\item We show how the cost expression for~each residual~join~is written in~a~simple~and~effective way~(Section 5).
\squishend
\section{Decomposition wrto HH}
First we need some definitions.

\vspace*{-.3cm}
 For each attribute  $X_i$ we define a set $L_{X_i}$  of {\em types}:
 \vspace*{-.3cm}
\squishlisttwob
\item
If $X_i$ has no values that are heavy hitters, then $L_{X_i}$ comprises of only one type, $T_{-}$,  called the {\em ordinary type}.\footnote{Ordinary type represents all other values of attribute $X_i$, the ones that are not heavy hitters.}
\item If $X_i$ has $p_i$ values that are heavy hitters, then $L_{X_i}$ comprises  of $1+p_i$ types:
one type $T_b$ for each heavy hitter, $b$, of $X_i$, and one ordinary type $T_{-}$.
\squishend
\vspace*{-.1cm}
A {\em combination of types}, $C_T$, is an element of the Cartesian product of the sets $L_{X_i}, i=1,2\ldots$ and defines a {\em residual join}.

 \vspace*{-.2cm}
E.g., for the query in in Example~\ref{2-way-ex}, we consider two residual joins, one for type combination $C_T=\{ A:T_{-}, B:T_{-},  C:T_{-} \}$ (without HH) and one for type combination $C_T=\{ A:T_{-}, B:T_{b},  C:T_{-} \}$ (with HH).

\vspace*{-.15cm}
Each $C_T$ defines a {\em residual join} which is the join computed only on a subset of the data.
Specifically, if an attribute $X$ has ordinary type in the current $C_T$ we  exclude the tuples for which $X=HH$. E.g.,  if there are two HH  $X=b_1$ and $X=b_2$, then we exclude (from all relations) all tuples with
 $X=b_1$ and $X=b_2$. If attribute $X$ is of type $T_b$ then we exclude (from all relations) the
tuples with value $X\neq b$.
\vspace*{-.3cm}
\begin{example}
\label{run1-ex}
We take as our running example the  3-way join:~~~ $J=R(A,B)\bowtie S(B,E,C) \bowtie T(C,D)$

\vspace*{-.15cm}
Suppose attribute $B$ has two HHs,  $B=b_1$ and $B=b_2$ and attribute $C$ has one HH,
and $C=c_1$. Thus attribute $B$ has three types, $T_{-}$, $T_{b_1}$ and $T_{b_2}$, attribute
$C$ has two types, $T_{-}$ and $T_{c_1}$ and the rest of the attributes have a single type, $T_{-}$.
Thus we have $3\times 2=6 $ residual joins, one for each combination.
By $r,s,t$ we denote the sizes of the relations that are {\em relevant} in each residual join, i.e., the number of tuples from each relation that
contribute in the particular residual join.  We list the residual joins:
\vspace*{-.2cm}
\squishlisttwo
\item All attributes  of type $T_{-}$. Here  $r$ is the number of only those tuples of relation $R$ for which
 $B\neq b_1$ and $B\neq b_2$, $s$ is the number of only those tuples of relation $S$ for which
 $B\neq b_1$ and $B\neq b_2$ and $C\neq c_1$, and $t$ is the number of those  tuples in relation $T$ for which $C\neq c_1$.
 \item All attributes of type $T_{-}$, except  $B$ of typle $T_{b_1}$.
\item All attributes  of type $T_{-}$, except  $B$  of typle $T_{b_2}$.
 \item All attributes  of type $T_{-}$, except  $C$ of typle $T_{c_1}$.

 \item All attributes a of type $T_{-}$, except  $B$ of type $T_{b_1}$ and  $C$  of typle $T_{c_1}$.
 \item All attributes  of type $T_{-}$, except  $B$ of type $T_{b_2}$ and  $C$  of typle $T_{c_1}$.

\squishend
\end{example}
\vspace*{-.3cm}
Each residual join is treated by the Shares algorithm as a separate join and a set of keys are defined that hash
each tuple as follows: A tuple $t$ of relation $R_j$ is sent to reducers of combinatrion $C_T$ only if the values of the tuple satisfy the constraints of $C_T$ as concerns values of HH.
\vspace*{-.3cm}
\begin{example}
\vspace*{-.2cm}
We continue from Example~\ref{run1-ex}.
Each tuple is sent to a number of reducers according to the keys created for each
residual join(we will provide more details later in the paper). E.g., a tuple $t$ from relation $R$ is sent to reducers as follows:
\vspace*{-.2cm}
\squishlisttwo
\item If $t$ has $B=b_1$ then it is sent to reducers created in items (2) and (5)  in Example~\ref{run1-ex}.
\item If $t$ has $B\neq b_1$  and $B\neq b_2$ then it is sent to reducers created in items (1) and (4).
\item If $t$ has $B=b_2$ then it is sent to reducers created in items  (3) and (6).
\squishend

\end{example}

\vspace*{-.5cm}
\section{Writing the Cost Expression}
\label{mod-shares-ex-sec}

In this section we will explain how to form a {\em HH-free residual join} and  how to compute the communication cost expression
for each residual join.
The structure we use in this section is conceptual, for the sake of showing how to write the cost expression. In practice, we do not materialize  $R(A,B_R)$ and $ S(B_S,C)$ or the auxiliary relation (definitions of these will be given shortly) -- as we
will explain in the next section. We begin with an example:
\vspace*{-.3cm}
\begin{example}
\label{res-2way-ex}
We  consider the residual join with HH $B=b$ for  the  join of Example~\ref{2-way-ex} which we rewrite here: $R(A,B) \bowtie S(B,C)$.
In order to do that,
we will equivalently imagine that we have to compute:
\squishlisttwo
 \item A join $R(A,B_R) \bowtie R_{aux} (B_R , B_S)
\bowtie S(B_S,C)$
\item On new database $D'$ which comes from $D$: We populate $R(A,B_R) $ with the same number of tuples as the original $R(A,B)$ has:
For each tuple $t$ with $B=b$, we add in  $R(A,B_R) $ a tuple where we have replaced the
$B=b$ with $B_R=b.t.R$. We do similarly for $S(B_S,C)$ replacing $B=b$ in tuple $t$ by
$B_S=b.t.S$.
The relation $R_{aux}(B_R,B_S)$ is the Cartesian product of $B_R$ and $B_S$.

\squishend
\squishlisttwob
\item {\em Observation 1:} Database $D'$ now has no heavy hitters. So, we can apply the original
Shares algorithm. The introduction of auxiliary attributes and relations may seem now as  if complicates things significantly, but, as we shall show in Section~\ref{dominance-sec}, it does not.
\squishend
\vspace*{-.3cm}
Thus if, in the database $D$, relation $R$ is $\{ (1,2),(3,2),(4,2)\}$ and $S$ is  $\{ (2,5),(2,6)\}$
then, in the database $D'$ we have (assuming $B=2$ qualifies for HH):\\\\
  $R(A,B_R) $ is $\{ (1,2.1.R),(3,2.3.R),(4,2.4.R)\}$.\\
$S(B_S,C)$ is  $\{ (2.5.S,5),(2.6.S,6)\}$.\\
(I.e., we conveniently identify the tuple of $R$ with the value of its first argument and the tuple of $S$ with the value of its second argument.)\\
The auxiliary relation $R_{aux}(B_R,B_S)$ is :\\ $\{ (2.1.R,2.5.S),(2.3.R,2.5.S),(2.4.R,2.5.S),$\\ $(2.1.R,2.6.S),(2.3.R,2.6.S),(2.4.R,2.6.S)\}$

The residual join computation has no heavy hitters, thus, we apply the original Shares algorithm, only that, when we compute the cost expression
 we ignore the communication cost for the auxiliary relation.\footnote{We can igonore it because we know what are the tuples in the
auxiliary relation and we can imagine that we can recreate them in the reducers. } Thus the communication cost of the
residual  join is again $ry + sx$, which is the same expression as in Example~\ref{2-way-ex}.
\end{example}
\vspace*{-.3cm}

The conceptual structure in the general case is as follows:
 For each combination of types, $C_T$,  we compute a HH-free residual join whose cost expression is written as follows:
\squishlisttwo
\item If attribute $X_i$  has non-ordinary type  in $C_T$ then:\\
--We introduce a number of
auxiliary attributes, one auxiliary attribute  for each relation $R_j$ where  attribute $X_i$ appears.  We denote the auxiliary
attribute for relation $R_j$ by  $X_{i-R_j}$.\\
--In the schema of each relation $R_j$ where $X_i$ appears, we replace  $X_i$ with  attribute $X_{i-R_j}$.

\item We form the residual join $J'$ for $C_T$ by   adding to original join new relations as follows: one relation, $R^{X_i}_{aux}$,
for each attribute $X_i$ which is not
of ordinary type. The schema of~ that~ relation~ consists of the attributes
$X_{i-R_j}$  for each $j$ such that $X_i$ is an attribute of $R_j$.

\item Now we write the communication cost expression for $J'$  as in the Shares algorithm taking care that:
\squishlisttwo
\item[a.] The communication cost expression does not include a term for auxiliary relations.
\item[b.] The size of each relation in $J'$ that we use in the cost expression is the number of tuples that
have as values in the arguments with heavy hitters the specific value for this combination of types.
\squishend
\squishend

Now we will discuss in the next subsection how (and why) to simplify the cost expression not to inclue
share variables for the auxiliary attributes.

\section{Dominance Relation: Its Role in Simplifying the Cost Expression}
\label{dominance-sec}

The property of the dominance relation allows us to write the cost expression for each residual join
in a simple manner. We use the theorem:

\vspace*{-.3cm}

\begin{theorem}
The share of each auxiliary attribute is equal to 1 in the optimum solution.\footnote{Sometimes, we have a tie where in a relation all attributes appear only once; in this case we break ties declaring always the auxiliary
attribute as dominated.}
\end{theorem}
\begin{proof}
  Each auxiliary attribute appears in one relation of the original join and in one auxiliary relation.
 Since we do not add a term in the cost expression for the auxiliary
relation,  we imagine that we write the cost expression for a join which is  the residual join without the auxiliary relations. Hence,
 an auxiliary attribute appears only in one relation,
hence it is dominated by a ordinary (non-HH in this residual join) attribute. There is the exception:  when all attributes in a relation are auxiliary attributes. In this case
there is only one tuple in the relation in this particular residual join, so  all attributes in the relation get share =1.
\end{proof}
\vspace*{-.3cm}
Thus we established that:
\squishlisttwob
\item The cost expression for each residual join can be derived from the cost expression of the original join (before dominance rule simplification) by making the shares of auxiliary attributes equal to 1.
\item Each tuple is hashed to reducers according to the values of the non-HH attributes in this tuple.
\squishend

\vspace*{-.5cm}
\begin{example}
We continue from Example~\ref{run1-ex} for the same HH as there.
 Remember by $a,b,c,d,e$ we denote the shares for each attribute $A,B,C,D,E$ respectively and by $r,s,t$ we denote the sizes of the relations that are {\em relevant} in each residual join, i.e., the number of tuples from each relation that
contribute in the particular residual join.
We always start with the cost expression for the original join, $r   cde +   s  ad +  t abe  $, and then simplify accordingly.
We list the cost expression for every residual join (and in the same order as) in Example~\ref{run1-ex}:
\squishlisttwo
 \item Here all attributes are ordinary, so we simplifly the relation by observing that
 $A$ is dominated by $B$ and $D$ is dominated by $C$, hence $a=1$ and $d=1$ and the expression is:

$r   c +   s   +  t b  $.

 \item Here $B$ is a non-ordinary attribute, hence $b=1$ and then, from the remaining attributes
 only $D$ is dominated by $C$,  hence $d=1$ and the expression is:
 $r   c +   s  a +  t a  $
  \item  $r   c +   s  a +  t a  $, i.e., same as above, only the sizes of the relations will be different.
 \item $r   d +   s  d +  t b  $
 \item Here we set both $b=1$ and $c=1$ and this gives us $r   de +   s  ad +  t ae  $.

\item $r   de +   s  ad +  t ae  $, i.e., same as above, only the sizes of the relations will be different.
\squishend

\end{example}

\bibliographystyle{abbrv}
\bibliography{skew-4-pages}
\end{document}